\documentclass[10pt]{article}

\usepackage[T1]{fontenc}										
\usepackage{lmodern}												
\usepackage{microtype}											
\usepackage[LGRgreek,frenchmath]{mathastext}
\usepackage[mathscr]{eucal}                 
\usepackage[sans]{dsfont}										
\usepackage{bbold}													
\usepackage{bbm}														
\usepackage{graphicx}                      	
\usepackage{mdwlist}												
\usepackage{natbib}    											
\setcitestyle{authoryear}
\usepackage[dvipsnames]{xcolor}							
\usepackage[plainpages=false, pdfpagelabels]{hyperref} 
	\hypersetup{
		colorlinks   = true,
		citecolor    = RoyalBlue, 
		linkcolor    = RubineRed, 
		urlcolor     = Turquoise
	}
\usepackage[paperwidth=8.5in,paperheight=11.00in,top=1.25in, bottom=1.25in, left=1.00in, right=1.00in]{geometry}
\usepackage{mathtools}                      
\mathtoolsset{showonlyrefs=true}            
\usepackage{amsmath}
\usepackage{amssymb}
\usepackage{amsfonts}
\usepackage{amsthm}                         
\allowdisplaybreaks                         
\newtheoremstyle{plain}
  {}   				
  {}   				
  {\itshape}  
  {}       		
  {\mdseries\scshape} 
  {.}         
  { } 				
  {\thmname{#1}\thmnumber{ #2}\ifx#3\empty\else\ (#3)\fi}
\theoremstyle{plain}
\newtheorem{theorem}{\underline{Theorem}}

\newtheoremstyle{definition}
  {}   				
  {}   				
  {}  				
  {}      		
  {\mdseries\scshape} 
  {.}         
  { } 				
  {\thmname{#1}\thmnumber{ #2}\ifx#3\empty\else\ (#3)\fi}
\theoremstyle{definition}
\newtheorem{definition}[theorem]{\underline{Definition}}
\newtheorem{example}[theorem]{\underline{Example}}
\newtheorem{remark}[theorem]{\underline{Remark}}

\newcommand{\<}{\langle}
\renewcommand{\>}{\rangle}

\renewcommand{\[}{\left[}
\renewcommand{\]}{\right]}

\newcommand\Eb{\mathds{E}}
\newcommand\Fb{\mathds{F}}

\newcommand\Pb{\mathds{P}}

\newcommand\Rb{\mathds{R}}

\newcommand\Nb{\mathds{N}}

\newcommand\Ac{\mathscr{A}}

\newcommand\Fc{\mathscr{F}}

\newcommand\Om{\Omega}
\newcommand\sig{\sigma}

\newcommand\gam{\gamma}
\newcommand\Gam{\Gamma}
\newcommand\lam{\lambda}
\newcommand\del{\delta}
\newcommand\Del{\Delta}
\renewcommand\phi{\varphi}

\newcommand\pb{\bar{p}}

\newcommand\Ebt{\widetilde{\Eb}}
\newcommand\Pbt{\widetilde{\Pb}}

\newcommand\Mt{\widetilde{M}}
\newcommand\Wt{\widetilde{W}}

\newcommand\Nt{\widetilde{N}}

\renewcommand\d{\partial}
\newcommand{\ind}{\perp \! \! \! \perp}
\newcommand\ii{\mathtt{i}}
\newcommand\dd{\mathrm{d}}
\newcommand\ee{\mathrm{e}}

\newcommand{\eg}{e.g.}

\begin{document}

\title{A primer on perpetuals}

\author{
Guillermo Angeris\thanks{Bain Capital Crypto. \textbf{e-mail}: \url{gangeris@baincapital.com}.}
\and
Tarun Chitra\thanks{Gauntlet Networks, Inc. \textbf{e-mail}: \url{tarun@gauntlet.network}.}
\and
Alex Evans\thanks{Bain Capital Crypto. \textbf{e-mail}: \url{aevans@baincapital.com}.}
\and
Matthew Lorig\thanks{Department of Applied Mathematics, University of Washington. \textbf{e-mail}: \url{mlorig@uw.edu}.}
}
\date{This version: \today}
\maketitle


\begin{abstract}
We consider a continuous-time financial market with no arbitrage and no transactions costs.
In this setting, we introduce two types of perpetual contracts, one in which the payoff to the long side is a fixed function of the underlyers and the long side pays a funding rate to the short side, the other in which the payoff to the long side is a fixed function of the underlyers times a discount factor that changes over time but no funding payments are required.  Assuming asset prices are continuous and strictly positive, we derive model-free expressions for the funding rate and discount rate of these perpetual contracts as well as replication strategies for the short side.  When asset prices can jump, we derive expressions for the funding and discount rates, which are semi-robust in the sense that they do not depend on the dynamics of the volatility process of the underlying risky assets, but do depend on the intensity of jumps under the market's pricing measure.   When asset prices can jump and the volatility process is independent of the underlying risky assets, we derive an explicit replication strategy for the short side of a perpetual contract.
Throughout the paper, we illustrate through examples how specific perpetual contracts relate to traditional financial instruments such as variance swaps and leveraged exchange traded funds.
\end{abstract}

%
%

\section{Introduction}
\label{sec:introduction}
A \textit{perpetual contract} (often just \emph{perp}, for short) is a type of
financial contract that enables relatively general payoffs. At a high level,
a perp contract can be described as follows: two parties, which we will call the
\textit{long side} and the \textit{short side}, enter into an agreement. The short side agrees to
pay the long side some payoff, which is a function of the prices of the underlying assets,
at a time of the long side's choosing. In exchange for this, the long side pays a continual
cash-flow to the short side up until contract termination. This cash
flow can be implemented in two distinct ways.  First, it can be implemented
directly as a literal cash flow, where the long side pays the short side cash
at fixed time increments.
Second, a sometimes more
practical approach is instead to replace the cash flow by discounting the payoff at contract termination.
\\[0.5em]
Perps were first suggested in~\cite{shiller1993} as a way of
approximately measuring the prices of dividend-yielding assets and also as a tool
to hedge certain illiquid assets.  But, only the special case where the
payoff function was linear in the price of the underlying assets was
considered. Perps with linear payoffs later gained widespread popularity as a way of taking
leveraged bets on cryptocurrency markets, where common derivatives markets
were initially relatively illiquid, if available at all. As of 2022,
perps are some of the most actively-traded cryptocurrency derivatives, with daily
volume in the tens of billions of dollars (see, \eg,~\cite{coingeckoperp}). 
\\[0.5em]
While perpetual futures never gained traction outside of cryptocurrencies, they were
introduced as a convenient way for miners (who produce newly minted tokens or coins) to
hedge inherent risks in cryptocurrency production.
The two main risks that miners of currencies such as Bitcoin and Ethereum face are that
their future income (which can be viewed as a dividend-yielding stream) is randomized with
variance depending on their resource contribution to the network~(\cite{lewenberg2015bitcoin}).
To reduce this variance, miners used two tactics: mining pools~(\eg~pooling together
resources and distributing dividends pro-rata) and futures contracts.
Initially, miner perpetual futures contracts were over-the-counter quanto futures, where
miners took premiums in cash with strike prices struck denominated in Bitcoin.
The failure of the early quanto derivatives market led to the creation of stablecoins
(dollar-pegged demand deposit assets) which then naturally led to the creation of the
crypto perpetuals market in 2016~(\cite{alexander2020bitmex}).
\\[0.5em]
Recently, perps with payoffs that are proportional to some power of
the price of an asset have been proposed (c.f.,~\cite{llllvvuu2021, white2021}). This extension
led to the creation of a decentralized perps protocol on the Ethereum
blockchain called Squeeth (short for ``Squared ETH''), see~\cite{squeeth2021}.
This protocol allows users to take long (and short) positions on these \textit{power
perps} without requiring an intermediary such as a broker or exchange.
\\[0.5em]
The role of this paper is to clearly define perps, show a number of natural
generalizations to those known in the literature, and correct some of the misinformation that exists online 
as to how the rate of cash-flow should be computed in a no-arbitrage setting.  
The rest of the paper proceeds as follows: in Section~\ref{sec:model} we introduce a financial market
in which risky assets have continuous price paths. Next, in
Section~\ref{sec:perp-1}, we define a perpetual contract in which the long side
must pay a funding rate to the short side.  We derive a model-free expression
for the funding rate as well as a replication strategy for the short side. In
Section~\ref{sec:perp-2} we define a second type of perpetual contract in which
no funding payments are required but the payoff is discounted over time. We derive a
model-free expression for the discount rate as well as a replication strategy for
the short side. Lastly, in Section~\ref{sec:jumps}, we consider a market with a
single risky asset whose value may jump.  In this setting, we derive
expressions for the funding and discount rates of the two types of perpetual
contracts introduced in Sections~\ref{sec:perp-1} and~\ref{sec:perp-2}.  And,
under the assumption of an independent volatility process, we derive a
replication strategy for the short side of a perpetual contract.

\section{Market model and assumptions}
\label{sec:model}
In Sections~\ref{sec:model}, \ref{sec:perp-1}, and~\ref{sec:perp-2}, we
consider a continuous-time financial market, defined on a filtered
probability space $(\Om,\Fc,\Fb,\Pb)$, with no arbitrage and no transaction
costs.  The filtration $\Fb = (\Fc_t)_{t \geq 0}$ represents the history of the
market and $\Pb$ represents the real-world probability measure.  We suppose
that the market contains a risk-free \textit{money market account}, whose value
in dollars is denoted by $M = (M_t)_{t \geq 0}$, as well as $n \in \Nb$
\textit{risky assets} (typically, tokens or cryptocurrencies), whose values in
dollars are denoted by
$S = (S_t^{(1)}, S_t^{(2)}, \ldots, S_t^{(n)})_{t \geq 0}$.
\\[0.5em]
We assume the value of the money market account $M$ is continuous, strictly positive and non-decreasing.  As such, there exists a non-negative $\Fb$-measurable process $r = (r_t)_{t \geq 0}$, known as the \textit{risk-free rate}, such that
\begin{align}
\dd M_t
	&=	r_t M_t \dd t , &
M_0
	&\geq 0 .
\end{align}
We further assume that the prices of the risky assets are continuous and strictly positive.  As such, there exists an $\Rb^n$-valued $\Fb$-measurable \textit{drift vector} $\mu = (\mu_t^{(1)},\mu_t^{(2)}, \ldots, \mu_t^{(n)})$ and an $\Rb_+^{n \times d}$-valued $\Fb$-measurable \textit{volatility matrix} $\sig = (\sig_t^{(1,1)}, \sig_t^{(1,2)}, \ldots, \sig_t^{(n,d)})$ with $d \in \Nb$, such that, for every $i$, the value $S^{(i)}$ of the $i$th risky asset is given by
\begin{align}
\dd S_t^{(i)}
	&=	\mu_t^{(i)} S_t^{(i)} \dd t + \sum_{j=1}^d \sig_t^{(i,j)} S_t^{(i)} \dd W_t^{(j)} , &
S_0^{(i)}
	&\geq	0 ,
\end{align}
where $W = (W_t^{(1)}, W_t^{(2)}, \ldots, W_t^{(d)})$ is a $d$-dimensional $(\Fb,\Pb)$-Brownian motion with independent components.
Lastly, throughout this paper, in order to avoid unnecessary complications, we assume all local martingales are true martingales.

\section{Perpetual contracts with funding}
\label{sec:perp-1}
We will discuss two types of perpetual contracts in this paper: (i) perpetual contracts with funding and (ii) perpetual contrast with discounting.  In this section, we focus on the former.  We begin with a definition.

\begin{definition}
\label{def:perp}
A \textit{perpetual contract with funding} (or simply, a \textit{perp}) written on $S$ with payoff function $\phi: \Rb^n \to \Rb$ is a agreement between two parties, referred to as the \textit{long side} and \textit{short side}.  The long side has the right to terminate the contract at any time $t \geq 0$, at which point it will receive a payment of $\phi(S_t)$.  In return, the long-side must pay to the short side $\phi(S_0)$ at the time $t=0$ of inception as well as a continuous $\Fb$-adapted cash-flow of $F=(F_t)_{t \geq 0}$ per unit time, referred to as the \textit{funding rate}, up until the contract is terminated.
\end{definition}


\noindent
The following theorem gives an expression for the funding rate.

\begin{theorem}
\label{thm:funding}
Consider a perpetual contract as described in Definition \ref{def:perp}.  Suppose that the function $\phi \in C^2(\Rb^d,\Rb)$.  Under the assumptions of Section \ref{sec:model}, the funding rate $F$ 
is given by
\begin{align}
F_t
	&=	\frac{1}{2} \sum_{i=1}^n \sum_{j=1}^n (\sig_t \sig_t^\top)^{(i,j)} S_t^{(i)} S_t^{(j)} \d_i \d_j \phi(S_t) 
			- \Big( \phi(S_t) - \sum_{i=1}^n S_t^{(i)} \d_i \phi(S_t) \Big) r_t .\label{eq:F}
\end{align}
where $(\sig_t \sig_t^\top)^{(i,j)}$ denotes the of $(i,j)$-th component of $\sig_t \sig_t^\top$ and $\d_j := \frac{\d}{\d s_j}$.
\end{theorem}

\begin{proof}
We will show that, with $F_t$ given by \eqref{eq:F}, the short side can create a self-financing portfolio whose value $X = (X_t)_{t \geq 0}$ satisfies 
\begin{align}
X_t 
	&=	\phi(S_t) , \label{eq:goal}
\end{align}
for all $t \geq 0$.  To begin, we note that the value of the short side's portfolio must be of the form
\begin{align}
\dd X_t 
	&=	\sum_{i=1}^n \Del_t^{(i)} \dd S_t^i + \Big( X_t - \sum_{i=1}^n \Del_t^{(i)} S_t^{(i)} \Big) \frac{1}{M_t} \dd M_t + F_t \dd t \\
	&=	\sum_{i=1}^n \Del_t^{(i)} \dd S_t^i + \Big( X_t - \sum_{i=1}^n \Del_t^{(i)} S_t^{(i)} \Big) r_t \dd t + F_t \dd t , \label{eq:dX}
\end{align}
where $\Del_t^{(i)}$ denotes the number of shares of invested in asset $i$ at time $t$.  Next, 
we have by It\^o's Lemma that
\begin{align}
\dd \phi(S_t)
	&=	\sum_{i=1}^n \d_i \phi(S_t) \dd S_t^{(i)} + \frac{1}{2} \sum_{i=1}^n \sum_{j=1}^n \d_i \d_j \phi(S_t) \dd \< S^{(i)}, S^{(i)} \>_t \\
	&=	\sum_{i=1}^n \d_i \phi(S_t) \dd S_t^{(i)} + \frac{1}{2} \sum_{i=1}^n \sum_{j=1}^n (\sig_t \sig_t^\top)^{(i,j)} S_t^{(i)} S_t^{(j)} \d_i \d_j \phi(S_t) \dd t . \label{eq:dphi}
\end{align}
Now, note that \eqref{eq:goal} holds for all $t \geq 0$ if and only if $X_0 = \phi(S_0)$ and 
$\dd X_t = \dd \phi(S_t)$.  
Comparing \eqref{eq:dX} with \eqref{eq:dphi}, we see that the $\dd S_t^{(i)}$ terms will be equal if we set
\begin{align}
\Del_t^{(i)}
	&=	\d_i \phi(S_t) . \label{eq:delta}
\end{align}
Next, using $X_t = \phi(S_t)$ and $\Del_t^{(i)} = \d_i \phi(S_t)$, we see that the $\dd t$ terms in \eqref{eq:dX} and \eqref{eq:dphi} will be equal if $F_t$ is given by \eqref{eq:F}.
\end{proof}

\begin{remark}
Note that the funding rate can be positive or negative.  If at time $t$ the funding rate is negative, then short side pays the long side at a rate of $-F_t$.
\end{remark}


\begin{remark}
Observe that precise knowledge of $r$, $\mu$ and $\sig$ is not needed to determine the funding rate $F_t$.  Indeed, using
\begin{align}
r_t
	&=	\frac{\dd }{\dd t} \log M_t, &
(\sig_t \sig_t^\top)^{(i,j)}
	&=	\frac{ \dd  }{\dd t} \< \log S^{(i)}, \log S^{(j)} \>_t,
\end{align}
we can express $F$ in the following \textit{model-free} form
\begin{align}
F_t \dd t
	&=	\frac{1}{2} \sum_{i=1}^n \sum_{j=1}^n S^{(i)} S^{(j)} \d_i \d_j \phi(S_t) \dd \< \log S^{(i)}, \log S^{(j)} \>_t
			- \Big( \phi(S_t) - \sum_{i=1}^n S_t^{(i)} \d_i \phi(S_t) \Big) \dd \log M_t . \label{eq:model-free-2}
\end{align}
By contrast, in order to price and replicate most traditional financial
derivatives such as European, American, Bermudan and Barrier options, one
requires a parametric model for the underlying $S$ as well as knowledge of
unobservable model parameters.
\end{remark}


\begin{example}
A (continuously monitored) \textit{variance swap} (VS), written on an asset $S \equiv S^{(1)}$ is an agreement between two parties, referred to as the \textit{long} and \textit{short} sides.  At the maturity date $T$, the short side pays the long side
\begin{align}
\int_0^T \dd \< \log S \>_t - K ,
\end{align}
where the \textit{swap rate} $K$ is determined at inception $t=0$ so that the
initial cost to enter the swap is zero.  Under the assumptions of Section
\ref{sec:model}, the swap rate $K$ is given by $- 2 \Ebt \log (S_T/S_0)$, where
$\Ebt$ denotes expectation under the market's chosen pricing measure $\Pbt$,
which can be deduced by observing implied volatilities of $T$-maturity European
calls and puts (see, e.g., \cite{carr2001towards}).  Because implied
volatilities tend to be higher than realized volatility (this is sometimes
known as the \textit{volatility premium}) taking the long side of a VS is
typically a losing trade.  As an alternative to entering the long side of a VS,
an investor wishing to gain exposure to volatility could take a long position
in a perp as described in Definition \ref{def:perp} with payoff $\phi(S_t) = 2
\log (S_t/S_0)$.  Like a VS, there is no cost to entering this perp because
$\phi(S_0) = 2 \log (S_0/S_0) = 0$.  Moreover, assuming $r \equiv 0$ for
simplicity, we have from \eqref{eq:model-free-2} that the funding rate is
\begin{align}
F_t \dd t
	&=	- \dd  \< \log S \>_t .
\end{align}
Therefore, if the long side chooses to terminate the contract at time $T$, the
value of the payoff minus funding paid is
\begin{align}
\phi(S_T) - \int_0^T F_t \, \dd t
	&=	2 \log \Big( \frac{S_T}{S_0} \Big) 
			+ \int_0^T \dd \<\log S \>_t .
\end{align}
Thus, by taking a long position in a perp, the investor can achieve the same exposure to volatility that they would have had they taken the long side of a VS, without paying a volatility premium.
\end{example}

\section{Perpetual contracts with discounting}
\label{sec:perp-2}
One of the problems with a perp with funding is that execution of the contract
requires the long side to place a deposit (e.g., on an exchange or into a
smart contract) at inception in order to pay the funding rate.
If the time-integral of the
funding rate ever exceeds the deposit, the contract is automatically
terminated.  One way to avoid automatic termination of the contract is to
consider, instead, a perpetual contract with discounting, whose mechanics are
described in the following definition.

\begin{definition}
\label{def:perp-2}
An \textit{perpetual contract with discounting} (or simply, a \textit{perp}) written on $S$ with payoff function $\phi: \Rb^n \to \Rb$ is a agreement between two parties, referred to as the \textit{long side} and \textit{short side}.  The long side has the right to terminate the contract at any time $t \geq 0$, at which point it will receive a payment of $\ee^{-\int_0^t D_s \dd s} \phi(S_t)$, where $D=(D_t)_{t \geq 0}$ is an $\Fb$-adapted process known as the \textit{discount rate}.  In return, at the time of inception $t=0$, the long side must pay to the short side a premium $\phi(S_0)$.
\end{definition}

\noindent
The following theorem gives an expression for the discount rate.

\begin{theorem}
\label{thm:decay}
Consider a perpetual contract with discounting as described in Definition \ref{def:perp-2}.  Suppose that the function $\phi \in C^2(\Rb^d,\Rb)$ and is either strictly positive or strictly negative.  Then, under the assumptions of Section \ref{sec:model}, the discount rate $D$ is given by
\begin{align}
D_t
	&=		\frac{F_t}{\phi(S_t)} , \label{eq:D}
\end{align}
where $F_t$ is given by \eqref{eq:F}.
\end{theorem}

\begin{proof}
We will show that, with $D_t$ given by \eqref{eq:D}, the short side can create a self-financing portfolio whose value $X = (X_t)_{t \geq 0}$ satisfies 
\begin{align}
X_t = \exp(-\int_0^t D_s \dd s)\phi(S_t) , \label{eq:matching}
\end{align}
for all $t \geq 0$.  To begin, we note that the dynamics of the short side's portfolio $X$ must be of the form
\begin{align}
\dd X_t
	&=	\sum_{i=1}^n \Del_t^{(i)} \dd S_t^i + \Big( X_t - \sum_{i=1}^n \Del_t^{(i)} S_t^{(i)} \Big) \frac{1}{M_t} \dd M_t \\
	&=	\sum_{i=1}^n \Del_t^{(i)} \dd S_t^i + \Big( X_t - \sum_{i=1}^n \Del_t^{(i)} S_t^{(i)} \Big) r_t \dd t , \label{eq:dX-2}
\end{align}
where $\Del_t^{(i)}$ denotes the number of shares of invested in asset $i$ at time $t$.  Next, we have by It\^o's Lemma that
\begin{align}
&\dd \Big( \ee^{- \int_0^t D_s \dd s} \phi(S_t) \Big) \\
	&=	\ee^{- \int_0^t D_s \dd s} 
			\Big( \sum_{i=1}^n \d_i \phi(S_t) \dd S_t^{(i)} + \frac{1}{2} \sum_{i=1}^n \sum_{j=1}^n \d_i \d_j \phi(S_t) \dd \< S^{(i)}, S^{(i)} \>_t \Big) 
			- D_t \ee^{- \int_0^t D_s \dd s}  \phi(S_t) \dd t \\
	&=	\ee^{- \int_0^t D_s \dd s} \Big( \sum_{i=1}^n \d_i \phi(S_t) \dd S_t^{(i)} 
			+ \frac{1}{2} \sum_{i=1}^n \sum_{j=1}^n (\sig_t \sig_t^\top)^{(i,j)} S_t^{(i)} S_t^{(j)} \d_i \d_j \phi(S_t) \dd t \Big) 
			- D_t \ee^{- \int_0^t D_s \dd s}  \phi(S_t) \dd t . \label{eq:dphi-2}
\end{align}
Now, note that \eqref{eq:matching} will hold
for all $t \geq 0$ if and only if $X_0 = \phi(S_0)$ and $\dd X_t = \dd( \ee^{-\int_0^t D_s \dd s} \phi(S_t))$.  Comparing \eqref{eq:dX-2} with \eqref{eq:dphi-2}, we see that the $\dd S_t^{(i)}$ terms will be equal if we set
\begin{align}
\Del_t^{(i)}
	&=	\ee^{- \int_0^t D_s \dd s} \d_i \phi(S_t) . \label{eq:delta-2}
\end{align}
Next, using $X_t = \exp(-\int_0^t D_s \dd s) \phi(S_t)$ and $\Del_t^{(i)} = \exp(-\int_0^t D_s \dd s) \d_i \phi(S_t)$, we see that the $\dd t$ terms in \eqref{eq:dX-2} and \eqref{eq:dphi-2} will be equal if $D_t$ is given by \eqref{eq:D}.
\end{proof}

\begin{example}
Let $L = (L_t)_{t \geq 0}$ be the value of a \textit{leveraged exchange traded fund} (LETF) with underlyer $S \equiv S^{(1)}$ and \textit{leverage ratio} $\gamma$.  The manager of such an LETF seeks to multiply the returns of $S$ by a factor of $\gamma$ by holding $\gam L_t / S_t$ shares of the underlyer for all $t \geq 0$ and borrowing from the bank to finance the position (see, e.g., \cite{leung2016leveraged}).  Thus, the dynamics of $L$ are as follows
\begin{align}
\dd L_t
	&=	\Del_t \dd S_t + (X_t - \Del_t S_t) \frac{1}{M_t} \dd M_t , &
\Del_t
	&=	\gam L_t / S_t .
\end{align}
Solving for $L_t$ one obtains the following expression
\begin{align}
L_t
	&=	L_0 \Big( \frac{S_t}{S_0} \Big)^\gam \Big( \frac{M_t}{M_0} \Big)^{(1-\gam)} \exp \Big( \frac{\gam(1-\gam)}{2} \< \log S \>_t \Big) .
\end{align}
Now, consider a perp as described in Definition \ref{def:perp-2} with payoff $\phi(S_t) = L_0 (S_t/S_0)^\gam$.  We have from \eqref{eq:model-free-2} and \eqref{eq:D} that
\begin{align}
D_t \dd t
	&=	\frac{\gam (\gam-1)}{2} \dd \< \log S \>_t - (1 - \gam) \dd \log M_t .
\end{align}
If the long side terminates the perp at time $t$ it will receive
\begin{align}
\phi(S_t) \ee^{- \int_0^t D_s \dd s}
	&=	L_0 \Big( \frac{S_t}{S_0} \Big)^\gam \exp \Big( -\frac{\gam (\gam-1)}{2} \int_0^t \dd \< \log S \>_s + (1 - \gam) \int_0^t \dd \log M_s \Big) \\
	&=	L_0 \Big( \frac{S_t}{S_0} \Big)^\gam \Big( \frac{M_t}{M_0} \Big)^{(1-\gam)} \exp \Big( \frac{\gam(1-\gam)}{2} \< \log S \>_t \Big) .
\end{align}
Thus, an LETF written on $S$ with leverage ratio $\gam$ can be viewed as a special case of a perp with payoff function $\phi(S_t) = L_0 (S_t/S_0)^\gam$.
Such perps trade widely on the Ethereum blockchain.  Indeed, \textit{Squared ETH} or \textit{Squeeth}, which trades on the Decentralized Finance (DeFi) protocol Opyn, is simply a perp with payoff function $\phi(S_t) = S_t^2$, where $S_t$ is the value on dollars of Ethereum.
\end{example}

\begin{example}
If, at time $t=0$, one deposits two tokens into a geometric mean constant function market maker (CFMM), then, ignoring fees collected by the CFMM, the value $V=(V_t)_{t \geq 0}$ of this deposit at time $t \geq 0$ is
\begin{align}
V_t
	&=	\Big( \frac{p V_0}{S_0^{(1)}} S_t^{(1)} + \frac{q V_0}{S_0^{(2)}} S_t^{(2)} \Big)
			\frac{(S_t^{(1)}/S_0^{(1)})^p (S_t^{(2)}/S_0^{(2)})^q}{p(S_t^{(1)}/S_0^{(1)}) + q(S_t^{(2)}/S_0^{(2)})} \\
	&=	V_0 \Big( \frac{S_t^{(1)}}{S_0^{(1)}} \Big)^p \Big( \frac{S_t^{(2)}}{S_0^{(2)}} \Big)^q , &
\end{align}
where $p$ and $q$ are constants satisfying $p,q>0$ and $p+q=1$ (c.f.,~\cite{AC20,evans2020liquidity}).

Now, consider a perp as described in Definition \ref{def:perp-2} with payoff 
\begin{align}
\phi(S_t) 
	&=	V_0 \Big( \frac{S_t^{(1)}}{S_0^{(1)}} \Big)^p \Big( \frac{S_t^{(2)}}{S_0^{(2)}} \Big)^q . \label{eq:gmm-payoff}
\end{align}
We have from \eqref{eq:model-free-2} and \eqref{eq:D} that
\begin{align}
D_t \dd t
	&=	\frac{p(p-1)}{2} \dd \< \log S^{(1)} \>_t + \frac{q(q-1)}{2} \dd \< \log S^{(2)} \>_t 
			+ p q \dd \< \log S^{(1)}, \log S^{(2)} \>_t .
\end{align}
If the perp is terminated at time $t$ the value of the payoff to the long side is
\begin{align}
\phi(S_t) \ee^{-\int_0^t D_s \dd s}
	&=	V_0 \Big( \frac{S_t^{(1)}}{S_0^{(1)}} \Big)^p \Big( \frac{S_t^{(2)}}{S_0^{(2)}} \Big)^q \\ &\quad
			\times \exp \Big(\frac{p(1-p)}{2} \< \log S^{(1)} \>_t + \frac{q(1-q)}{2} \< \log S^{(2)} \>_t - p q \< \log S^{(1)}, \log S^{(2)} \>_t \Big) .
\end{align}
One can show that the term in the exponent is positive along every path of $(S^{(1)},S^{(2)})$.
\footnote{This follows from $p (1-p) \dd x^2 + q (1-q) \dd y^2 - 2 p q \dd x \dd y =	p q \dd x^2 + p q \dd y^2 - 2 p q \dd x \dd y =	( p \dd x - q \dd y)^2 \geq 0$.}
Thus, ignoring fees, rather than deposit tokens into a CFMM it would always be better to take a long position in a perp with payoff \eqref{eq:gmm-payoff}.
\end{example}

\section{Extension to models with jumps}
\label{sec:jumps}
In this section, we derive funding and discount rates for perpetual contracts as well as replication strategies for the short side when asset prices are allowed to jump.  For simplicity, we will assume the risk-free rate of interest is zero ($r_t = 0$) and we will consider perpetuals written on a single risky asset $S$.  The extension to non-zero interest rates and multiple assets it straightforward, but tedious.
\\[0.5em]
Fix a filtered probability space $(\Om,\Fc,\Fb,\Pbt)$ where the filtration $\Fb$ represents the history of the market and 
$\Pbt$ denotes the \textit{market's chosen pricing measure}.  We suppose that the dynamics of the risky asset $S$ are of the form
\begin{align}
\dd S_t
	&=	\sig_{t-} S_{t-} \dd \Wt_t + \int S_{t-} ( \ee^{\gam_{t-}(z)} - 1 ) \Nt(\dd t, \dd z) , \label{eq:dS-jumps-4}
\end{align}
where $\sig = (\sig_t)_{t \geq 0}$ and $\gam(z) = (\gam_t(z))_{t \geq 0}$ for every $z \in \Rb$ are scalar $\Fb$-adapted processes, $\Wt = (\Wt_t)_{t \geq 0}$ is a scalar $(\Pbt,\Fb)$-Brownian motion and $\Nt(\dd t, \dd z)= N(\dd t, \dd z) - \nu(\dd z) \dd t$ is a compensated Poisson random measure on $\Rb$.  Observe that $S$ is a $(\Pbt,\Fb)$-martingale, as it must be in the absence of arbitrage.  The following theorem gives the funding rate $F$ and discount rate $D$ for the perpetual contracts described in Definitions \ref{def:perp} and \ref{def:perp-2}, respectively, when the dynamics of the underlying $S$ are given by \eqref{eq:dS-jumps-4}.

\begin{theorem}
\label{thm:F-and-D}
Suppose the dynamics of a single underlying risky asset are of the form \eqref{eq:dS-jumps-4} and the payoff function $\phi$ of a perpetual satisfies $\phi \in C^2(\Rb,\Rb)$ and $\phi \neq 0$.  Then, the funding rate $F$ of the perpetual contract described in Definition \ref{def:perp} is given by
\begin{align}
F_{t}
	&=	\frac{1}{2} \sig_{t}^2 S_{t}^2 \phi''(S_{t}) 
			+ \int \Big( \phi(S_{t} \ee^{\gam_{t}(z)}) - \phi(S_{t}) - S_{t} ( \ee^{\gam_{t}(z)} - 1)\phi'(S_{t}) \Big) \nu(\dd z) . \label{eq:F-jumps}
\end{align}
and the discount rate $D$ of a perpetual contract described in Definition \ref{def:perp-2} is given by $D_{t} = F_{t}/\phi(S_{t})$.
\end{theorem}

\begin{proof}
Consider first a perpetual contract with funding, as described in Definition \ref{def:perp}.  The infinitesimal change in value of the long side is
\begin{align}
\dd \phi(S_t) - F_{t} \dd t
	&=	\Ac_{t} \phi(S_{t}) \dd t + \sig_{t} S_{t} \phi'(S_{t}) \dd \Wt_t 
			+ \int \Big( \phi(S_{t-} \ee^{\gam_{t-}(z)} ) - \phi(S_{t-}) \Big) \Nt(\dd t, \dd z) - F_{t} \dd t, \label{eq:d-long} \\
\Ac_{t} \phi(S_{t})
	&:=	\frac{1}{2} \sig_{t}^2 S_{t}^2 \phi''(S_{t}) 
			+ \int \Big( \phi(S_{t} \ee^{\gam_{t}(z)} ) - \phi(S_{t}) - S_{t}(\ee^{\gam_{t}(z)} - 1) \phi'(S_{t}) \Big) \nu(\dd z)  .
\end{align}
In the absence of arbitrage, the value of the long-side must be a $(\Pbt,\Fb)$-martingale.  As such, the sum of the $\dd t$ terms in \eqref{eq:d-long} must be zero, which leads to the expression \eqref{eq:F-jumps} for $F$.  Next, consider a perpetual contract with discounting, as described in Definition \ref{def:perp-2}.  The infinitesimal change in value of the long side is given by
\begin{align}
\dd \Big( \ee^{-\int_0^t D_s \dd s} \phi(S_t) \Big)
	&=	- D_t \ee^{-\int_0^t D_s \dd s} \phi(S_t) \dd t + \ee^{-\int_0^t D_s \dd s} \dd \phi(S_t) \\
	&=	- D_t \ee^{-\int_0^t D_s \dd s} \phi(S_t) \dd t + \ee^{-\int_0^t D_s \dd s} \Ac_t \phi(S_t) \dd t  
			+ \ee^{-\int_0^t D_s \dd s} \sig_t S_t \phi''(S_t) \dd \Wt_t \\ & \quad
			+ \ee^{-\int_0^t D_s \dd s} \int \Big( \phi(S_{t-} \ee^{\gam_{t-}(z)} ) - \phi(S_{t-}) \Big) \Nt(\dd t, \dd z) . \label{eq:d-long-2}
\end{align}
Once again, in the absence of arbitrage, the value of the long side must be a $(\Pbt,\Fb)$-martingale.  As such, the sum of the $\dd t$ terms in \eqref{eq:d-long-2} must be zero, which leads to $D_{t} = F_{t}/\phi(S_{t})$, where $F$ is given by \eqref{eq:F-jumps}.
\end{proof}

\begin{remark}
Note that, unlike the proofs of Theorems \ref{thm:funding} and \ref{thm:decay}, which provide expressions for the funding rate $F$ and discount rate $D$ of the perpetual contracts described in Definitions \ref{def:perp} and \ref{def:perp-2} \textit{as well as} replication strategies for the short-side, the proof of Theorem \ref{thm:F-and-D} provides only the funding and discount rates for perpetual contracts but says nothing about a replication strategy for the short-side.  In order to derive a replication strategy for the short side when the underlying asset $S$ can jump, we will need to make some additional assumptions.
\end{remark}


\noindent
Henceforth, assume that the dynamics of the risky asset $S$ are of the form
\begin{align}
\dd S_t
	&=	\sig_{t-} S_{t-} \dd \Wt_t + \int S_{t-} ( \ee^z - 1 ) \Nt(\dd t, \dd z) , \label{eq:dS-jumps-3}
\end{align}
where the L\'evy measure $\nu$ associated with $\Nt$ is a \textit{Dirac comb}
\begin{align}
\nu(\dd z)
	&=	\sum_{j=1}^n \lam_j \del_{z_j}(z) \dd z , \label{eq:comb}
\end{align}
and the volatility process $\sig$ evolves independently of the Brownian motion $\Wt$ and the Poisson random measure $\Nt$ that appear in \eqref{eq:dS-jumps-3}
\begin{align}
\sig
	&\ind \Wt , &
\sig
	&\ind \Nt . \label{eq:sigma-assumptions}
\end{align}
Assume further that one can trade call and put options with any strike at a fixed maturity $T$.  
This is equivalent to assuming one can trade any $T$-maturity European option whose payoff can be written as the difference of convex functions, as these payoffs can be synthesized from call and put payoffs.

\begin{theorem}
\label{thm:replication}
Suppose the dynamics of $S$ satisfy \eqref{eq:dS-jumps-3}, \eqref{eq:comb} and \eqref{eq:sigma-assumptions}
and consider a perpetual contract as described in Definition \ref{def:perp}.
From Theorem \ref{thm:F-and-D} the funding rate $F$ is given by
\begin{align}
F_{t}
	&=	\frac{1}{2} \sig_{t}^2 S_{t}^2 \phi''(S_{t}) 
			+ \int \Big( \phi(S_{t} \ee^z) - \phi(S_{t}) - S_{t} ( \ee^z - 1)\phi'(S_{t}) \Big) \nu(\dd z) . \label{eq:F-3}
\end{align}
Let $P^{(p)} = (P_t^{(p)})_{0 \leq t \leq T}$ denote the value of a European power contract, which pays $S_T^p$ at time $T$ where $p \in \Rb$.
Fix $(p_1, p_2, \ldots, p_n) \in \Rb^n$ such that the $(n + 1) \times (n+1)$ stochastic matrix $H = (H_t)_{0 \leq t \leq T}$ with entries given by
\begin{align}
\left. \begin{aligned}
H_{t-}^{(j,i)}
	&=	\ee^{\psi(\pb_i) (T-t)} P_{t-}^{(p_i)} S_{t-}^{\pb_i} ( \ee^{p_i z_j} - 1 ) - \ee^{\psi(p_i) (T-t)} P_{t-}^{(\pb_i)} S_{t-}^{p_i}  ( \ee^{\pb_i z_j} - 1 ) , &
j,i
	&\leq n , \\
H_{t-}^{(j,n+1)}
	&=	S_{t-} (\ee^{z_j} - 1 ), &
j
	&\leq n , \\
H_{t-}^{(n+1,i)}
	&=  \sig_{t-} \Big( p_i \ee^{\psi(\pb_i) (T-t)} P_{t-}^{(p_i)} S_{t-}^{\pb_i}  - \pb_i \ee^{\psi(p_i) (T-t)} P_{t-}^{(\pb_i)} S_{t-}^{p_i} \Big) , &
i
	&\leq n , \\
H_{t-}^{(n+1,n+1)}
	&= \sig_{t-} S_{t-} . 
\end{aligned} \right\} \label{eq:H}
\end{align}
is invertible for all $t \in [0,T]$, where $\pb_i := 1-p_i$ for all $i$ and the function $\psi$ is defined as follows
\begin{align}
\psi(p)
	&=	\int \Big( (\ee^{pz} - 1) - p(\ee^z-1) \Big) \nu(\dd z) .
\end{align}
Let $X = (X_t)_{0 \leq t \leq T}$ be a the value of self-financing portfolio with dynamics of the form
\begin{align}
\dd X_t
	&=	\Del_{t-} \dd S_t + \sum_{i=1}^n \Gam_{t-}^{(p_i)} \dd Y_t^{(p_i)} + F_{t-} \dd t , &
X_0
	&=	\phi(S_0) , \label{eq:dX-jumps} \\
\dd Y_t^{(p_i)}
	&=	\ee^{\psi(\pb_i) (T-t)} S_{t-}^{\pb_i} \dd P_t^{(p_i)} -  \ee^{\psi(p_i) (T-t)} S_{t-}^{p_i} \dd P_t^{(\pb_i)} , \label{eq:dY-def} 
\end{align}
where $\Del$ and $\Gam^{(p_1)}, \ldots, \Gam^{(p_n)}$ are given by
\begin{align}
\[ \begin{array}{c}
\Gam_{t-}^{(p_1)} \\
\vdots \\
\Gam_{t-}^{(p_n)} \\
\Del_{t-}
\end{array} \] 
	&=	H_{t-}^{-1}
\[ \begin{array}{c}
\phi( S_{t-} \ee^{p z_1})  - \phi(S_{t-}) \\
\vdots \\
\phi( S_{t-} \ee^{p z_n})  - \phi(S_{t-}) \\
\sig_{t-} S_{t-} \phi'(S_{t-}) 
\end{array} \] . \label{eq:coeffs}
\end{align}
Then the portfolio $X$ 
replicates the perpetual payoff.  That is, the following holds
\begin{align}
X_t 
	&=	\phi(S_t) , \label{eq:replication}
\end{align}
for all $0 \leq t \leq T$.
\end{theorem}

\begin{proof}
We begin by computing the dynamics of $\phi(S)$.  Using \eqref{eq:dS-jumps-3} and It\^o's Lemma we have
\begin{align}
\dd \phi(S_t)
	&=	\tfrac{1}{2} \sig_{t-}^2 S_{t-}^2 \phi''(S_{t-}) \dd t + \sig_{t-} S_{t-} \phi'(S_{t-}) \dd \Wt_t \\ &\quad
			+	\int \Big( \phi(S_{t-}\ee^z) - \phi(S_{t-}) - S_{t-}(\ee^z-1)\phi'(S_{t-}) \Big) \nu(\dd s) \dd t \\ &\quad
			+ \int \Big( \phi(S_{t-}\ee^z) - \phi(S_{t-}) \Big) \Nt(\dd t, \dd z) . \label{eq:dSp-3}
\end{align}
Next, it will be helpful to introduce $Z = \log S$.  We will separate $Z$ into a \textit{continuous} component $Z^c$ and a \textit{jump} component $Z^j$. Using It\^o's Lemma, we have
\begin{align}
\dd Z_t
	&=	\dd Z_t^c + \dd Z_t^j , &
\dd Z_t^c
	&=	-\tfrac{1}{2} \sig_{t-}^2 \dd t + \sig_{-t} \dd \Wt_t , &
\dd Z_t^j
	&=	- \int ( \ee^z - 1 - z ) \nu( \dd z) \dd t + \int z \Nt(\dd t, \dd z) .
\end{align}
Note that $Z_T^c-Z_t^c$ is normally distributed conditional on the path of $\sig$ and that $Z^j$ is a L\'evy process with characteristic exponent $\psi(\ii \, \cdot \,)$.  Thus, conditioning on the path of $\sig$ and using the L\'evy-Kintchine formula, we have
\begin{align}
\Ebt_t \ee^{p(Z_T^c - Z_t^c)}
	&=	\Ebt_t \ee^{\tfrac{1}{2}(p^2-p)(\<Z^c\>_T-\<Z^c\>_t ) } , &
\<Z^c\>_T-\<Z^c\>_t
	&=	\int_t^T \sig_s^2 \dd s ,  &
\Ebt_t \ee^{p(Z_T^j - Z_t^j)}
	&=	\ee^{(T-t) \psi(p)} ,  \label{eq:Zc-char}
\end{align}
where we have introduced the short-hand notation $\Ebt_t \, \cdot \, := \Ebt( \, \cdot \, | \Fc_t )$.  Now, 
using \eqref{eq:Zc-char} 
as well as $Z^c \ind Z^j$, we find that the value of a European power option satisfies
\begin{align}
P_t^{(p)}
	&=	\Ebt_t S_T^p
	=		\ee^{p Z_t} \Ebt_t \ee^{p(Z_T^c-Z_t^c)} \Ebt_t \ee^{p(Z_T^j-Z_t^j)} 
	=		S_t^p \ee^{\psi(p) (T-t)} \Ebt_t \ee^{\tfrac{1}{2}(p^2-p) (\<Z^c\>_T-\<Z^c\>_t )} . 
\end{align}
Next, observe that
\begin{align}
\dd \Ebt_t \ee^{\tfrac{1}{2}(p^2-p) (\<Z^c\>_T-\<Z^c\>_t )}
	&=	( \ldots ) \dd t + \dd \Mt_t^{(p)} ,
\end{align}
where $\Mt^{(p)}$ is a $(\Pbt,\Fb)$-martingale, which is independent of $\Wt$ and $\Nt$.  Thus, using the fact that $P^{(p)}$ is a $(\Pbt,\Fb)$-martingale, and thus $\dd P_t^{(p)}$ has no $\dd t$ terms, we have
\begin{align}
\dd P_t^{(p)}
	&=	\ee^{\psi(p) (T-t)} \Ebt_{t-} \ee^{\tfrac{1}{2}(p^2-p) (\<Z^c\>_T-\<Z^c\>_t )} \dd S_t^p  
			+ S_{t-}^p \ee^{\psi(p) (T-t)} \dd \Ebt_t \ee^{\tfrac{1}{2}(p^2-p) (\<Z^c\>_T-\<Z^c\>_t)} \\ & \quad
			+ S_{t-}^p \Ebt_t \ee^{\tfrac{1}{2}(p^2-p) (\<Z^c\>_T-\<Z^c\>_t)} \dd \ee^{\psi(p) (T-t)} \\
	&=	\ee^{\psi(p) (T-t)} \Ebt_{t-} \ee^{\tfrac{1}{2}(p^2-p) (\<Z^c\>_T-\<Z^c\>_t)} 
			\Big( p \sig_{t-} S_{t-}^{p} \dd \Wt_t + \int S_{t-}^{p} ( \ee^{p z} - 1 ) \Nt(\dd t, \dd z) \Big) \\&\quad
			+ S_{t-}^p \ee^{\psi(p) (T-t)} \dd \Mt^{(p)} \\
	&=	p \sig_{t-} P_{t-}^{(p)} \dd \Wt_t + \int P_{t-}^{(p)} ( \ee^{p z} - 1 ) \Nt(\dd t, \dd z) 
			+ \ee^{\psi(p) (T-t)}  S_{t-}^p \dd \Mt_t^{(p)} . \label{eq:dPp}
\end{align}
Now, defining $\pb := 1-p$ and noting that $p^2 - p = \pb^2 - \pb$, we have $\Mt^{(p)} = \Mt^{(\pb)}$.  Therefore, we have from \eqref{eq:dPp} that
\begin{align}
\dd P_t^{(\pb)}
	&=	\pb \sig_{t-} P_{t-}^{(\pb)} \dd \Wt_t + \int P_{t-}^{(\pb)} ( \ee^{\pb z} - 1 ) \Nt(\dd t, \dd z) 
			+ \ee^{\psi(\pb) (T-t)} S_{t-}^{\pb} \dd \Mt_t^{(p)} . \label{eq:dPpbar}
\end{align}
Thus, using \eqref{eq:dPp} and \eqref{eq:dPpbar}, the process $Y^{(p)}$, defined in \eqref{eq:dY-def}, is a self-financing portfolio that satisfies
\begin{align}
\dd Y_t^{(p)}
	&=	\ee^{\psi(\pb) (T-t)} S_{t-}^{\pb} \Big( p \sig_{t-} P_{t-}^{(p)} \dd \Wt_t + \int P_{t-}^{(p)} ( \ee^{p z} - 1 ) \Nt(\dd t, \dd z) \Big) \\ & \quad
			- \ee^{\psi(p) (T-t)} S_{t-}^p \Big( \pb \sig_{t-} P_{t-}^{(\pb)} \dd \Wt_t + \int P_{t-}^{(\pb)} ( \ee^{\pb z} - 1 ) \Nt(\dd t, \dd z) \Big) . \label{eq:dYp}
\end{align}
Now, we wish to create a self-financing portfolio whose value $X$ satisfies \eqref{eq:replication}.
As there are at least $(n+1)$ sources of uncertainty (due to the Brownian motion $\Wt$ and the $n$ possible jump sizes $\Del Z_t^j \in \{z_1, z_2, \ldots, z_n\}$) the portfolio will need at least $(n+1)$ hedging assets; we will use 
the underlying $S$ as well as ``shares'' of $Y^{(p_i)}$ for $i \in \{1,2,\ldots,n\}$.
Thus, the dynamics of $X$ are of the form \eqref{eq:dX-jumps}, 
where, for the moment, the processes $\Del$ and $\Gam^{(p_1)}, \ldots, \Gam^{(p_n)}$ are unknown.
Using \eqref{eq:dS-jumps-3}, \eqref{eq:dX-jumps} and \eqref{eq:dYp}, we have
\begin{align}
\dd X_t 
	&=	\Del_{t-} \sig_{t-} S_{t-} \dd \Wt_t + \sum_{i=1}^n \Gam_{t-}^{(p_i)} 
			\sig_{t-} \Big( p_i \ee^{\psi(\pb_i) (T-t)} P_{t-}^{(p_i)} S_{t-}^{\pb_i}  - \pb_i \ee^{\psi(p_i) (T-t)} P_{t-}^{(\pb_i)} S_{t-}^{p_i} \Big) \dd \Wt_t \\ &\quad
			+ \int \Del_{t-} S_{t-} (\ee^z - 1 ) \Nt(\dd t, \dd z) +  \\ & \quad
			+ \int \sum_{i=1}^n \Gam_{t-}^{(p_i)} \Big( 
			\ee^{\psi(\pb_i) (T-t)} P_{t-}^{(p_i)} S_{t-}^{\pb_i} ( \ee^{p_i z} - 1 ) - \ee^{\psi(p_i) (T-t)} P_{t-}^{(\pb_i)} S_{t-}^{p_i}  ( \ee^{\pb_i z} - 1 )
			\Big) \Nt(\dd t, \dd z)  
			+ F_{t-} \dd t .  \label{eq:dX-3}
\end{align}
Equation \eqref{eq:replication} will be satisfied if and only if the $\dd t$, $\dd \Wt_t$ and $\Nt(\dd t, \dd z)$ terms in \eqref{eq:dSp-3} and \eqref{eq:dX-3} are equal.  As such, the funding rate $F$ must be given by \eqref{eq:F-3} and the processes $\Del$ and $\Gam^{(p_1)}, \ldots, \Gam^{(p_n)}$ must satisfy
\begin{align}
\sig_{t-} S_{t-} \phi'(S_{t-}) 
	&=	\Del_{t-} \sig_{t-} S_{t-} + \sum_{i=1}^n \Gam_{t-}^{(p_i)} 
			\sig_{t-} \Big( p_i \ee^{\psi(\pb_i) (T-t)} P_{t-}^{(p_i)} S_{t-}^{\pb_i}  - \pb_i \ee^{\psi(p_i) (T-t)} P_{t-}^{(\pb_i)} S_{t-}^{p_i} \Big)  \\
\phi(S_{t-}\ee^{z_j}) - \phi(S_{t-})
	&= \Del_{t-} S_{t-} (\ee^{z_j} - 1 ) \\ & \quad 
			+ \sum_{i=1}^n \Gam_{t-}^{(p_i)} \Big( 
			\ee^{\psi(\pb_i) (T-t)} P_{t-}^{(p_i)} S_{t-}^{\pb_i} ( \ee^{p_i z_j} - 1 ) - \ee^{\psi(p_i) (T-t)} P_{t-}^{(\pb_i)} S_{t-}^{p_i}  ( \ee^{\pb_i z_j} - 1 )
			\Big) ,
\end{align}
where the last equation must hold for all $z_j \in \{z_1, z_2, \ldots, z_n\}$.  In matrix form, we have
\begin{align}
\[ \begin{array}{c}
\phi(S_{t-}\ee^{z_1}) - \phi(S_{t-}) \\
\vdots \\
\phi(S_{t-}\ee^{z_n}) - \phi(S_{t-}) \\
\sig_{t-} S_{t-} \phi'(S_{t-})
\end{array} \]
	&=
H_{t-}
\[ \begin{array}{c}
\Gam_{t-}^{(p_1)} \\
\vdots \\
\Gam_{t-}^{(p_n)} \\
\Del_{t-}
\end{array} \] ,
\end{align}
where the entries of $H$ are given by \eqref{eq:H}.
Using the fact that $H$ is invertible, we find that \eqref{eq:replication} will hold if $\Del$ and $\Gam^{(p_1)}, \ldots, \Gam^{(p_n)}$ are given by \eqref{eq:coeffs}.
\end{proof}

\begin{remark}
The replication strategy described in Theorem \ref{thm:replication} works only up until the maturity date $T$ of the European power contracts.  However, at time $T$ one can continue the replication strategy by trading European power contracts with a maturity date $\overline{T} > T$.
\end{remark}

\begin{remark}
Note that $S_{t-}$, $P_{t-}^{(p_i)}$ for all $i$ and $\sig_{t-} = \frac{\dd}{\dd t} \<Z^c\>_t$ are observable.  Thus, no assumptions about the dynamics of the volatility process $\sig$ are needed for the replication strategy to work.  We do, however, require knowledge of the possible jump-sizes $\{ z_1, z_2, \ldots, z_n\}$ and jump intensities under the pricing measure $\{\lam_1, \lam_2, \ldots, \lam_n \}$ as these appear in $\nu$ and $\psi$.
\end{remark}



\bibliography{references}


\end{document}